\documentclass{aims}

\usepackage{txfonts}
\usepackage{doi}   
\usepackage{url}   
\def\typeofarticle{Research article}
\def\currentvolume{8}
\def\currentissue{x}
\def\currentyear{2024}

\def\ppages{xxx--xxx}
\def\DOI{10.3934/math.2024xxx}
\def\Received{January 2024}
\def\Revised{February 2024}
\def\Accepted{March 2024}
\def\Published{March 2024}

\numberwithin{equation}{section}

\makeatletter
\renewcommand{\@biblabel}[1]{#1\hfill \hspace{-0.2cm}}
\makeatother

\newtheorem{theorem}{Theorem}[section]
\newtheorem{lemma}[theorem]{Lemma}
\newtheorem{proposition}[theorem]{Proposition}
\newtheorem{corollary}[theorem]{Corollary}
\theoremstyle{definition}
\newtheorem{definition}[theorem]{Definition}
\newtheorem{example}[theorem]{Example}
\newtheorem{remark}[theorem]{Remark}

\newcommand{\F}{\mathbb{F}}
\newcommand{\Tr}{\mathrm{Tr}}
\newcommand{\card}{\operatorname{card}}

\newcommand{\Div}{\operatorname{Div}}
\newcommand{\charac}{\operatorname{char}}
\newcommand{\R}{\mathbb{R}}

\begin{document}

\title{Additive Conjucyclic Codes over $\F_{q^2}$: A Trace Correspondence and Quantum Error-Correction}

\author{%
  Jingjie Lv\affil{1}\corrauth,
  Xian Lian\affil{1},
    Ruihu Li\affil{2},
  and
  Hanxu Hou\affil{1,3}
}

\shortauthors{the Author(s)}

\address{%
  \addr{\affilnum{1}}{School of Electrical Engineering \& Intelligentization, Dongguan University of Technology, Dongguan 523808, China}
  \addr{\affilnum{2}}{Fundamentals Department, Air Force   Engineering University, Xi'an 710051, China}
\addr{\affilnum{3}}{Institute of Network Coding, The Chinese University of Hong Kong, Hong Kong, 999077, China}}

\corraddr{Email: juxianljj@163.com
}

\begin{abstract}
Additive conjucyclic codes over $\F_{q^2}$ are closed under the conjugated cyclic shift and play an important role in constructing quantum error-correcting codes (QECCs). However, a systematic algebraic theory for such codes over general finite fields has been lacking. In this paper, we develop a unified framework by establishing a trace-based $\F_q$-linear isomorphism between $\F_{q^2}^n$ and $\F_q^{2n}$. This correspondence shows that additive conjucyclic codes of length $n$ correspond bijectively to $q$-ary linear cyclic codes of length $2n$, translating their structural analysis to the well-understood setting of cyclic codes. Using this isomorphism, we determine the enumeration of such codes and give explicit forms of their generator matrices.
We then introduce an alternating inner product on $\F_{q^2}^n$, which is shown to be compatible with the symplectic inner product on $\F_q^{2n}$ under the trace isomorphism. Based on this inner product, we characterize the dual-containing condition for additive conjucyclic codes and derive explicit parity-check matrices. Finally, we construct $q$-ary QECCs from dual-containing additive conjucyclic codes.
Our results unify and generalize previous studies on quaternary additive conjucyclic codes and present a construction method for $q$-ary QECCs from additive conjucyclic codes, together with an illustrative example.
\end{abstract}


\maketitle

\section{Introduction}
Additive codes were first introduced by Delsarte and Levenshtein in 1973~\cite{delsarte1973}. A seminal contribution by Calderbank \textit{et al.}~\cite{calderbank1998} established a direct connection between quaternary additive codes and binary quantum error-correcting codes (QECCs). The latter play a crucial role in mitigating environmental and operational decoherence in quantum information, thereby underpinning both quantum communication and quantum computation. Since then, additive codes have attracted sustained research interest, with recent advances spanning their structural properties, optimal parameters, and applications in quantum error correction~\cite{abdukhalikov2026, bhowmick2026, benbelkacem2020, ezerman2011, huffman2013, kim2003, kim2017, kurz2026, mahmoudi2019}.

Additive cyclic codes constitute a distinguished subclass of additive codes, characterized by invariance under cyclic shifts (see Definition~\ref{def:cyclic}). This structural property facilitates efficient encoding and decoding via shift registers. The algebraic structure of quaternary additive cyclic codes was first investigated by Huffman~\cite{huffman2007}, who established a canonical decomposition for codes of odd length. Subsequently, Bierbrauer~\cite{bierbrauer2012} extended the theory to additive cyclic codes over arbitrary finite fields. G\"uneri \textit{et al.}~\cite{guneri2017} derived a sufficient condition under which a particular family of additive cyclic codes are complementary dual. The existence of asymptotically good additive cyclic codes was later demonstrated by Shi \textit{et al.}~\cite{shi2018}. More recently, Dastbasteh and Shivji~\cite{dastbasteh2023} introduced a polynomial representation for additive cyclic codes over $\F_{p^2}$ and derived new binary quantum codes from this framework. Very recently, further developments include the characterization of trace duality and additive complementary pairs for additive cyclic codes over finite chain rings~\cite{bhowmick2026}.

Similar to additive cyclic codes, additive conjucyclic codes are closed under the conjugated cyclic shift (see Definition~\ref{def:conjucyclic}) and can also be utilized to construct QECCs~\cite{calderbank1998}. However, in contrast to the extensive study of additive cyclic codes, research on additive conjucyclic codes remains remarkably limited. A major obstacle is that these codes do not admit a canonical polynomial representation, which complicates their structural analysis. Early studies on conjucyclic codes were confined to the linear setting, primarily over $\F_4$ \cite{abualrub2007} or related rings~\cite{sirap2008}. The additive case was first systematically addressed by Abualrub \textit{et al.}~\cite{abualrub2020, abualrub2022}, who developed an algebraic description over $\F_4$ using a linear algebraic approach. Nevertheless, their method relies heavily on the specific properties of $\F_4$ and does not readily generalize to arbitrary $q^2$-ary alphabets. To date, the literature on additive conjucyclic codes remains sparse, with no unified framework applicable to general finite fields $\F_{q^2}$. This gap motivates the present work, in which we aim to develop a systematic algebraic theory for additive conjucyclic codes over $\F_{q^2}$ by establishing a trace-based isomorphism with $q$-ary linear cyclic codes. This approach not only reveals the underlying cyclic structure but also facilitates the construction of QECCs via the alternating inner product.

In this paper, we investigate the algebraic structure of additive conjucyclic codes over $\F_{q^2}$. By employing the trace function from $\F_{q^2}$ down to $\F_q$, we establish an isomorphic correspondence between $q^2$-ary additive conjucyclic codes and $q$-ary linear cyclic codes. Using this correspondence, we determine the enumeration of such codes and provide explicit forms of their generator matrices. Furthermore, by defining an alternating inner product on $\F_{q^2}^n$, we characterize the dual-containing condition for these codes and derive explicit parity-check matrices, which then yield a construction method for $q$-ary quantum error-correcting codes from additive conjucyclic codes.

The remainder of this paper is organized as follows. Section~\ref{sec:prelim} introduces the necessary definitions and preliminary results on cyclic codes, conjucyclic codes, as well as the Euclidean and symplectic inner products. A list of the main symbols used throughout the paper is provided at the end of this section. In Section~\ref{sec:structure}, we establish a trace-based correspondence between $q^2$-ary additive conjucyclic codes and $q$-ary linear cyclic codes, and give their enumeration and generator matrices. Section~\ref{sec:generator} defines the alternating inner product on $\F_{q^2}^n$, investigates the dual structure under this inner product, derives explicit parity-check matrices, and presents a construction method for $q$-ary QECCs from additive conjucyclic codes, together with an illustrative example. Section~\ref{sec:conclusion} concludes the paper.

\section{Preliminaries}\label{sec:prelim}
Let $\F_{q^{2}}$ be a finite field with $q^{2}$ elements, where $q$ is a prime power $p$. Denote by $\charac(\F_{q^{2}}) = p$ the characteristic of $\F_{q^{2}}$. For any element $\alpha \in \F_{q^{2}}$, let $\bar{\alpha} = \alpha^{q}$ denote its conjugation. A linear code with parameters $[n,k]_{q^{2}}$ is a $k$-dimensional linear subspace of $\F_{q^{2}}^{n}$. An additive code $\mathcal{C}$ of length $n$ over $\F_{q^{2}}$ is a subgroup of the additive group of $\F_{q^{2}}^{n}$; if $|\mathcal{C}| = M$, we refer to it as an $(n,M)_{q^{2}}$ additive code. In general, for an arbitrary element $k \in \F_{q^{2}}$ and a codeword $\mathbf{c} \in \mathcal{C}$, the scalar multiple $k\mathbf{c}$ need not lie in $\mathcal{C}$. When $\mathcal{C}$ is closed under multiplication by elements of $\F_{q}$ (i.e., $k\mathbf{c} \in \mathcal{C}$ for all $k \in \F_{q}$ and $\mathbf{c} \in \mathcal{C}$), we call $\mathcal{C}$ an $\F_{q}$-linear additive code. In this paper, the term additive code will always refer to an $\F_{q}$-linear additive code unless stated otherwise.

For a vector $\pmb{u} = (u_0, u_1, \ldots, u_{n-1}) \in \F_q^n$, define its Hamming weight as $w_h(\pmb{u}) = \card\{i \mid u_i \neq 0,\ 0 \le i \le n-1\}$. For a $q$-ary linear or additive code $\mathcal{C}$ of length $n$, the minimum Hamming weight is given by
\[
w_h(\mathcal{C}) = \min \{ w_h(\pmb{u}) \mid \pmb{u} \in \mathcal{C},\ \pmb{u} \neq \mathbf{0} \}.
\]
If $n = 2m$, for $0 \le j \le m-1$, define the symplectic weight of $\pmb{u}$ as $w_s(\pmb{u}) = \card\{j \mid (u_j, u_{m+j}) \neq (0,0)\}$. The minimum symplectic weight of $\mathcal{C}$ is then defined as
\[
w_s(\mathcal{C}) = \min \{ w_s(\pmb{u}) \mid \pmb{u} \in \mathcal{C},\ \pmb{u} \neq \mathbf{0} \}.
\]

We now introduce the definition of cyclic codes.

\begin{definition}\label{def:cyclic}
	Let $\mathcal{C}$ be a linear or additive code of length $n$ over $\F_q$. If for every codeword $\mathbf{c} = (c_0, c_1, \ldots, c_{n-1}) \in \mathcal{C}$, its right cyclic shift
	\begin{equation}\label{eq111xx}
		\sigma(\mathbf{c}) = (c_{n-1}, c_0, \ldots, c_{n-2})
	\end{equation}
 also belongs to $\mathcal{C}$, then $\mathcal{C}$ is called a $q$-ary cyclic code of length $n$.
\end{definition}

Let $\R = \F_q[x] / \langle x^n - 1 \rangle$ be the quotient ring. Define an $\F_q$-module isomorphism $\phi: \F_q^n \to \R$ by
\[
\phi(c_0, c_1, \ldots, c_{n-1}) = c_0 + c_1 x + \cdots + c_{n-1} x^{n-1}.
\]
It is easy to see that $\mathcal{C}$ is a $q$-ary linear cyclic code of length $n$ if and only if $\phi(\mathcal{C})$ is an ideal of $\R$. Since $\R$ is a principal ideal ring, there exists a one-to-one correspondence between monic divisors of $x^n - 1$ in $\F_q[x]$ and $q$-ary linear cyclic codes of length $n$. Thus, the factorization of $x^n - 1$ is essential for classifying such codes. If a linear cyclic code $\mathcal{C}$ is generated by a monic divisor $g(x)$ of $x^n - 1$, we call $g(x)$ the generator polynomial of $\mathcal{C}$; for convenience, we denote $\mathcal{C}$ by $\langle g(x) \rangle$ throughout.

The definition of conjucyclic codes is analogous.

\begin{definition}\label{def:conjucyclic}
	Let $\mathcal{C}$ be a $q^2$-ary linear or additive code of length $n$. For any codeword $\pmb{c} = (c_0, c_1, \ldots, c_{n-1}) \in \mathcal{C}$, if its right conjucyclic shift 
	\begin{equation}\label{eq222xx}
		T(\pmb{c}) = (\bar{c}_{n-1}, c_0, \ldots, c_{n-2})
	\end{equation}
	lies in $\mathcal{C}$, then $\mathcal{C}$ is called a $q^2$-ary conjucyclic code of length $n$.
\end{definition}

We now introduce several inner products and their associated dual codes, which will be used throughout this paper.

For two vectors $\pmb{u} = (u_0, u_1, \ldots, u_{n-1})$ and $\pmb{v} = (v_0, v_1, \ldots, v_{n-1})$ in $\F_q^n$, the Euclidean inner product is defined as
\[
\langle \pmb{u}, \pmb{v} \rangle_e = \sum_{i=0}^{n-1} u_i v_i.
\]
If $n = 2m$, the symplectic inner product is given by
\begin{equation}\label{wewe12}
\langle \pmb{u}, \pmb{v} \rangle_s = \sum_{i=0}^{m-1} (u_i v_{m+i} - u_{m+i} v_i).
\end{equation}
The corresponding Euclidean and symplectic dual codes of a code $\mathcal{C}$ are respectively defined as
\begin{equation*}
\begin{split}
\mathcal{C}^{\perp_e} = \{ \pmb{v} \in \F_q^n \mid \langle \pmb{u}, \pmb{v} \rangle_e = 0,\ \forall \pmb{u} \in \mathcal{C} \},~~
\mathcal{C}^{\perp_s} = \{ \pmb{v} \in \F_q^{2m} \mid \langle \pmb{u}, \pmb{v} \rangle_s = 0,\ \forall \pmb{u} \in \mathcal{C} \}.
\end{split}
\end{equation*}
	
	\section*{List of symbols}
For the reader's convenience, we list the main symbols used in this paper:
	\begin{center}
		\begin{tabular}{l l}
			$\mathbb{F}_{q^2}$ & finite field with $q^2$ elements, $q$ a prime power $p$ \\
			$\bar{\alpha} = \alpha^q$ & conjugation in $\mathbb{F}_{q^2}$ \\
			$\Tr: \mathbb{F}_{q^2} \to \mathbb{F}_q$ & trace map, $\Tr(x)=x+x^q$ \\
			$\beta$ & a fixed primitive element of $\mathbb{F}_{q^2}$ \\
			$\mathcal{C}^{\perp_e},\mathcal{C}^{\perp_s},\mathcal{C}^{\perp_a}$ & Euclidean, symplectic and alternating dual codes of $\mathcal{C}$ \\
		$\sigma$ & right cyclic shift operator, $\sigma(c_0,\dots,c_{n-1}) = (c_{n-1},c_0,\dots,c_{n-2})$ \\
			$T$ & right conjucyclic shift operator, $T(c_0,\dots,c_{n-1})=(\bar{c}_{n-1},c_0,\dots,c_{n-2})$ \\
			$\varphi_\beta: \mathbb{F}_{q^2}\to\mathbb{F}_q^2$ & $\varphi_\beta(\alpha)=(\Tr(\beta\alpha),\Tr(\bar{\beta}\alpha))$ \\
			$\Psi_\beta: \mathbb{F}_{q^2}^n\to\mathbb{F}_q^{2n}$ & componentwise extension of $\varphi_\beta$ \\
			$\langle\cdot,\cdot\rangle_e$ & Euclidean inner product \\
			$\langle\cdot,\cdot\rangle_s$ & symplectic inner product on $\mathbb{F}_q^{2m}$ \\
			$\langle\cdot,\cdot\rangle_{a,\beta}$ & alternating inner product on $\mathbb{F}_{q^2}^n$ \\
			$\Div_{\F_q}(x^{2n} - 1)$& set of divisors of $x^{2n}-1$ over $\F_q$\\
			$g(x)$ & generator polynomial of a cyclic code (divisor of $x^{2n}-1$) \\
			$h(x)$ & $(x^{2n}-1)/g(x)$\\
			$h^*(x)$ & reciprocal polynomial of $h(x)$ \\
			$\pmb{f}$ & coefficient vector of polynomial $f(x)$ \\
			$\tau$ & linear transformation $\tau(v_0,\dots,v_{2n-1})=(-v_n,\dots,-v_{2n-1},v_0,\dots,v_{n-1})$ \\
			$\hat{G}$ & generator matrix of an additive conjucyclic code $\mathcal{C}$ \\
			$\hat{H}_a$ & generator matrix of $\mathcal{C}^{\perp_{a,\beta}}$ (also parity-check matrix of $\mathcal{C}$) \\
		\end{tabular}
	\end{center}

\section{Algebraic correspondence and generator description}\label{sec:structure}

In this section, we first show that there exists a bijective correspondence between $q^2$-ary additive conjucyclic codes of length $n$ and $q$-ary linear cyclic codes of length $2n$.

Let $\Tr: \F_{q^2} \to \F_q$ denote the trace map defined by $\Tr(x) = x + x^q$ for all $x \in \F_{q^2}$. For a fixed primitive element $\beta \in \F_{q^2}$, define a mapping $\varphi_\beta: \F_{q^2} \to \F_q^2$ by
\begin{equation}\label{eq333xx}
\varphi_\beta(\alpha) = \bigl( \Tr(\beta\alpha),\; \Tr(\bar{\beta}\alpha) \bigr).
\end{equation}

\begin{proposition}
	The map $\varphi_\beta$ in \eqref{eq333xx} is an $\F_q$-linear isomorphism.
\end{proposition}
\begin{proof}
	For any $k_1, k_2 \in \F_q$ and $\alpha_1, \alpha_2 \in \F_{q^2}$, we have
	\[
	\begin{aligned}
		\varphi_\beta(k_1\alpha_1 + k_2\alpha_2)
		&= \bigl( \Tr(\beta(k_1\alpha_1 + k_2\alpha_2)),\; \Tr(\bar{\beta}(k_1\alpha_1 + k_2\alpha_2)) \bigr) \\
		&= \bigl( k_1\Tr(\beta\alpha_1) + k_2\Tr(\beta\alpha_2),\; k_1\Tr(\bar{\beta}\alpha_1) + k_2\Tr(\bar{\beta}\alpha_2) \bigr) \\
		&= k_1\varphi_\beta(\alpha_1) + k_2\varphi_\beta(\alpha_2),
	\end{aligned}
	\]
	so $\varphi_\beta$ is $\F_q$-linear.
	
	To prove injectivity, assume $\varphi_\beta(\alpha_1) = \varphi_\beta(\alpha_2)$. Then
	\[
	\Tr(\beta(\alpha_1 - \alpha_2)) = 0, \quad\Tr(\bar{\beta}(\alpha_1-\alpha_2)) = 0,
	\]
	which is equivalent to the system
	\begin{equation}\label{eq1}
		\begin{cases}
		\beta(\alpha_1 - \alpha_2) + \bar{\beta}(\alpha_1 - \alpha_2)^q = 0\\
		\bar{\beta}(\alpha_1 - \alpha_2) + \beta(\alpha_1 - \alpha_2)^q = 0
	\end{cases}.
	\end{equation}
	If $\charac(\F_{q^2}) = 2$, then \eqref{eq1} gives
	\[
	\beta(\alpha_1 + \alpha_2) = \bar{\beta}(\alpha_1 + \alpha_2)^q,\quad 
	\bar{\beta}(\alpha_1 + \alpha_2) = \beta(\alpha_1 + \alpha_2)^q,
	\]
	so both $\beta(\alpha_1 + \alpha_2)$ and $\bar{\beta}(\alpha_1 + \alpha_2)$ lie in $\F_q$. Since $\beta$ is primitive, this forces $\alpha_1 = \alpha_2$.
	
	Now suppose $\charac(\F_{q^2}) \neq 2$. Adding and subtracting the two equations in \eqref{eq1} yields
	\[
	\begin{cases}
		(\beta + \bar{\beta})\bigl((\alpha_1 - \alpha_2) + (\alpha_1 - \alpha_2)^q\bigr) = 0\\
		(\beta - \bar{\beta})\bigl((\alpha_1 - \alpha_2) - (\alpha_1 - \alpha_2)^q\bigr) = 0
	\end{cases}.
	\]
	Note that $\beta + \bar{\beta} \neq 0$ and $\beta - \bar{\beta} \neq 0$; otherwise $\beta^{2q-2}=1$ or $\beta^{q-1}=1$, contradicting the fact that $\beta$ has order $q^2-1$ with $q-1 < 2q-2 < q^2-1$. Hence
	\[
	\begin{cases}
		(\alpha_1 - \alpha_2) + (\alpha_1 - \alpha_2)^q = 0\\
		(\alpha_1 - \alpha_2) - (\alpha_1 - \alpha_2)^q = 0
	\end{cases},
	\]
	which implies $2(\alpha_1 - \alpha_2) = 0$ and therefore $\alpha_1 = \alpha_2$ (since $p \neq 2$). Thus $\varphi_\beta$ is injective, and because both spaces have dimension $2$ over $\F_q$, it is an isomorphism.
\end{proof}

\begin{remark}\label{rem:inverse}
	Notice that the mapping $\varphi_\beta$ is an $\F_q$-linear isomorphism. Hence any element $\alpha \in \F_{q^2}$ is uniquely determined by $(\Tr(\beta\alpha), \Tr(\bar{\beta}\alpha)) \in \F_q^2$. In fact,
	\begin{equation}\label{eq11111}
		\alpha = \frac{1}{\beta - \beta^{2q-1}} \Tr(\beta\alpha) - \frac{\beta^{q-1}}{\beta - \beta^{2q-1}} \Tr(\bar{\beta}\alpha).
	\end{equation}
\end{remark}

\begin{example}\label{ex:q3}
	Let $q = 3$ and let $\beta$ be a primitive element of $\F_9$. Then the mapping $\varphi_\beta: \F_9 \to \F_3^2$ is given as follows:
	\[
	\renewcommand{\arraystretch}{0.8}
	\begin{array}{c|c}
		\alpha & \varphi_\beta(\alpha) \\ \hline
		0 & (0,0) \\
		\beta^0 & (1,1) \\
		\beta & (0,1) \\
		\beta^2 & (1,2) \\
		\beta^3 & (1,0)
	\end{array}
\hspace{0.7cm}\vrule\hspace{0.7cm}
	\begin{array}{c|c}
		\alpha & \varphi_\beta(\alpha) \\ \hline
		\beta^4 & (2,2) \\
		\beta^5 & (0,2) \\
		\beta^6 & (2,1) \\
		\beta^7 & (2,0)
	\end{array}~.
	\]
 According to \eqref{eq11111} in Remark~\ref{rem:inverse}, we compute
	\[
	\frac{1}{\beta - \beta^5} = \beta^3,\quad \frac{\beta^{2}}{\beta - \beta^{5}} = \beta^5
	\]
	and the identities can be summarized as follows, where \((a,b) = \varphi_\beta(\alpha) = \bigl(\Tr(\beta\alpha), \Tr(\bar{\beta}\alpha)\bigr)\):
	\[
	\renewcommand{\arraystretch}{0.8}
	\begin{array}{c|c}
		\beta^3 \cdot a - \beta^5 \cdot b & \alpha \\ \hline
		\beta^3 \cdot 0 - \beta^5 \cdot 0 & 0 \\
		\beta^3 \cdot 1 - \beta^5 \cdot 1 & \beta^0 \\
		\beta^3 \cdot 0 - \beta^5 \cdot 1 & \beta \\
		\beta^3 \cdot 1 - \beta^5 \cdot 2 & \beta^2 \\
		\beta^3 \cdot 1 - \beta^5 \cdot 0 & \beta^3
	\end{array}
	\hspace{0.7cm}\vrule\hspace{0.7cm}
	\begin{array}{c|c}
		\beta^3 \cdot a - \beta^5 \cdot b & \alpha \\ \hline
		\beta^3 \cdot 2 - \beta^5 \cdot 2 & \beta^4 \\
		\beta^3 \cdot 0 - \beta^5 \cdot 2 & \beta^5 \\
		\beta^3 \cdot 2 - \beta^5 \cdot 1 & \beta^6 \\
		\beta^3 \cdot 2 - \beta^5 \cdot 0 & \beta^7
	\end{array}~.
	\]
\end{example}

The map $\varphi_\beta$ extends componentwise to an $\F_q$-linear isomorphism $\Psi_\beta: \F_{q^2}^n \to \F_q^{2n}$ as follows: for any vector $(\alpha_0, \alpha_1, \ldots, \alpha_{n-1}) \in \F_{q^2}^n$,
\begin{equation*}
	\Psi_\beta(\alpha_0, \alpha_1, \ldots, \alpha_{n-1}) = \bigl( \varphi_\beta(\alpha_0), \varphi_\beta(\alpha_1), \ldots, \varphi_\beta(\alpha_{n-1}) \bigr),
\end{equation*}
where each $\varphi_\beta(\alpha_i) = (\Tr(\beta\alpha_i), \Tr(\bar{\beta}\alpha_i))$ is regarded as a pair in $\F_q^2$. Expanding the pairs in the order of all first components followed by all second components, we obtain the explicit coordinate representation
\begin{equation}\label{eq2-exp}
	\Psi_\beta(\alpha_0,\alpha_1,\ldots,\alpha_{n-1}) = \bigl(\Tr(\beta\alpha_0),\Tr(\beta\alpha_1),\ldots,\Tr(\beta\alpha_{n-1}),\Tr(\bar{\beta}\alpha_0),\ldots,\Tr(\bar{\beta}\alpha_{n-1})\bigr).
\end{equation}

\begin{proposition}\label{prop:commute}
	For any vector $\alpha \in \F_{q^2}^n$ and any integer $i \ge 0$, we have $\Psi_\beta(T^i(\alpha)) = \sigma^i(\Psi_\beta(\alpha))$, where $\sigma$ and $T$ are the cyclic shift and conjucyclic shift operators defined in \eqref{eq111xx} and \eqref{eq222xx}, respectively. In other words, the following diagram commutes for all $i$:
	\[
	\begin{array}{ccc}
		\F_{q^2}^n & \xrightarrow{T^i} & \F_{q^2}^n\\
		\downarrow\Psi_\beta && \downarrow\Psi_\beta\\
		\F_q^{2n} & \xrightarrow{\sigma^i} & \F_q^{2n}.
	\end{array}
	\]
\end{proposition}
\begin{proof}
	We first prove the case $i=1$. For any $\alpha = (\alpha_0, \ldots, \alpha_{n-1}) \in \F_{q^2}^n$, we compute
	\[
	\begin{aligned}
		\Psi_\beta(T(\alpha)) &= \Psi_\beta(\bar{\alpha}_{n-1}, \alpha_0, \ldots, \alpha_{n-2}) \\
		&= \bigl( \Tr(\beta\bar{\alpha}_{n-1}), \Tr(\beta\alpha_0), \ldots, \Tr(\beta\alpha_{n-2}), 
	\Tr(\bar{\beta}\bar{\alpha}_{n-1}), \Tr(\bar{\beta}\alpha_0), \ldots, \Tr(\bar{\beta}\alpha_{n-2}) \bigr),
	\end{aligned}
	\]
	\[
	\begin{aligned}
		\sigma(\Psi_\beta(\alpha)) &= \sigma\bigl( \Tr(\beta\alpha_0), \ldots, \Tr(\beta\alpha_{n-1}), \Tr(\bar{\beta}\alpha_0), \ldots, \Tr(\bar{\beta}\alpha_{n-1}) \bigr) \\
		&= \bigl( \Tr(\bar{\beta}\alpha_{n-1}), \Tr(\beta\alpha_0), \ldots, \Tr(\beta\alpha_{n-2}),  \Tr(\beta\alpha_{n-1}), \Tr(\bar{\beta}\alpha_0), \ldots, \Tr(\bar{\beta}\alpha_{n-2}) \bigr).
	\end{aligned}
	\]
	Using the identities $\Tr(\beta\bar{\alpha}_{n-1}) = \Tr(\bar{\beta}\alpha_{n-1})$ and $\Tr(\bar{\beta}\bar{\alpha}_{n-1}) = \Tr(\beta\alpha_{n-1})$ (which follow directly from the definition of the trace), we see that the first component of $\Psi_\beta(T(\alpha))$ coincides with the first component of $\sigma(\Psi_\beta(\alpha))$, and the $(n+1)$-th component of $\Psi_\beta(T(\alpha))$ coincides with the $n$-th component of $\sigma(\Psi_\beta(\alpha))$. All remaining components are identical by construction. Hence $\Psi_\beta(T(\alpha)) = \sigma(\Psi_\beta(\alpha))$, which establishes the case $i=1$.
	
	Now assume that $\Psi_\beta(T^{i-1}(\alpha)) = \sigma^{i-1}(\Psi_\beta(\alpha))$ holds for some $i \ge 2$. Then
	\[
	\Psi_\beta(T^i(\alpha)) = \Psi_\beta(T(T^{i-1}(\alpha))) = \sigma(\Psi_\beta(T^{i-1}(\alpha))) = \sigma(\sigma^{i-1}(\Psi_\beta(\alpha))) = \sigma^i(\Psi_\beta(\alpha)).
	\]
	Thus the statement holds for all $i \ge 0$ by induction.
\end{proof}

Consequently, we obtain a bijective correspondence between additive conjucyclic codes and linear cyclic codes.

\begin{theorem}\label{thm:correspondence}
	A code $\mathcal{C} \subseteq \F_{q^2}^n$ is an additive conjucyclic code if and only if $\mathcal{D} = \Psi_\beta(\mathcal{C})$ is a $q$-ary linear cyclic code of length $2n$. Moreover, the minimum Hamming weight of $\mathcal{C}$ equals the minimum symplectic weight of $\mathcal{D}$, i.e., $w_h(\mathcal{C}) = w_s(\mathcal{D})$.
\end{theorem}
\begin{proof}
	Suppose $\mathcal{C}$ is an additive conjucyclic code of length $n$ over $\F_{q^2}$. Then $\mathcal{D} = \Psi_\beta(\mathcal{C})$ is an $\F_q$-linear subspace of $\F_q^{2n}$, i.e., a $q$-ary linear code of length $2n$. By Proposition~\ref{prop:commute}, we have $\sigma(\mathcal{D}) = \sigma(\Psi_\beta(\mathcal{C})) = \Psi_\beta(T(\mathcal{C})) = \Psi_\beta(\mathcal{C}) = \mathcal{D}$, so $\mathcal{D}$ is cyclic.
	
	Conversely, assume $\mathcal{D}$ is a $q$-ary linear cyclic code of length $2n$. Since $\Psi_\beta$ is an $\F_q$-linear isomorphism, $\mathcal{C} = \Psi_\beta^{-1}(\mathcal{D})$ is an additive subgroup of $\F_{q^2}^n$. Moreover, $\Psi_\beta(T(\mathcal{C})) = \sigma(\Psi_\beta(\mathcal{C})) = \sigma(\mathcal{D}) = \mathcal{D} = \Psi_\beta(\mathcal{C})$, which implies $T(\mathcal{C}) = \mathcal{C}$. Hence $\mathcal{C}$ is an additive conjucyclic code.
	
	For any codeword $\mathbf{c} = (c_0, c_1, \ldots, c_{n-1}) \in \mathcal{C}$, it follows from the definition of $\Psi_\beta$ that $c_i = 0$ if and only if $(\Tr(\beta c_i), \Tr(\bar{\beta} c_i)) = (0,0)$. Consequently, $w_h(\mathcal{C}) = w_s(\mathcal{D})$.
\end{proof}

\begin{remark}
	By the theory of cyclic codes \cite{huffman2003}, the binary cyclic codes derived from additive conjucyclic codes over $\F_4$ in \cite{abualrub2020} are necessarily repeated-root, since $\gcd(2n,2)=2$. It is known that repeated-root cyclic codes are asymptotically inferior to simple-root cyclic codes \cite{castagnoli1991}. In contrast, the construction presented in this work yields $q$-ary cyclic codes that can be either repeated-root or simple-root, depending on the factorization of $x^{2n}-1$ over $\F_q$. This flexibility allows us to avoid the asymptotic limitations of repeated-root codes and opens the possibility of constructing quantum codes with better asymptotic parameters.
\end{remark}

Let $2n = p^{\ell}n_0$, where $p = \charac(\F_q)$ and $\ell \ge 0$, $n_0 > 0$ are integers such that $\gcd(n_0, p) = 1$. Assume that $x^{n_0} - 1$ factorizes into distinct irreducible polynomials in $\F_q[x]$ as
\[
x^{n_0} - 1 = g_1(x) g_2(x) \cdots g_t(x),
\]
where each $g_i(x)$ is irreducible and $g_i(x) \neq g_j(x)$ for $i \neq j$. Then
\begin{equation}\label{eq:factor}
	x^{2n} - 1 = (x^{n_0} - 1)^{p^{\ell}} = g_1(x)^{p^{\ell}} g_2(x)^{p^{\ell}} \cdots g_t(x)^{p^{\ell}}.
\end{equation}
The set of divisors of $x^{2n} - 1$ over $\F_q$ is therefore
\[
\Div_{\F_q}(x^{2n} - 1) = \left\{ g_1(x)^{s_1} g_2(x)^{s_2} \cdots g_t(x)^{s_t} \mid 0 \le s_1, s_2, \ldots, s_t \le p^{\ell} \right\}.
\]
Consequently, $x^{2n} - 1$ has $(p^{\ell} + 1)^t$ distinct monic divisors in $\F_q[x]$. By the isomorphism $\Psi_\beta$ defined in \eqref{eq2-exp}, this immediately yields the number of $q^2$-ary additive conjucyclic codes of length $n$.

\begin{theorem}\label{theo3.7}
	Assume that the polynomial $x^{2n}-1$ is factorized as Eq. \eqref{eq:factor}, then there exist $(p^{\ell}+1)^t$ distinct $q^2$-ary additive conjucyclic codes of length $n$.
\end{theorem}

Let $g(x) = g_1(x)^{s_1} \cdots g_t(x)^{s_t} \in \Div_{\F_q}(x^{2n}-1)$ with $\deg( g_i(x)) = d_i$. Then the cyclic code $\langle g(x) \rangle$ has dimension $2n - \sum_{i=1}^t s_i d_i$, so $|\langle g(x) \rangle| = q^{2n - \sum s_i d_i}$, which is also the size of the corresponding additive conjucyclic code.



We now turn to the description of generator matrices for additive conjucyclic codes. To this end, we first recall the notion of a generator matrix for $\F_q$-linear additive codes over $\F_{q^2}$.

\begin{definition}
	A generator matrix of an $\F_q$-linear additive code $\mathcal{C}$ over $\F_{q^2}$ is a matrix whose rows form an $\F_q$-basis of $\mathcal{C}$.
\end{definition}

For a polynomial
$
f(x) = f_0 + f_1x + \cdots + f_{2n-1}x^{2n-1} \in \F_q[x] / \langle x^{2n}-1 \rangle,
$
we denote by $\pmb{f} = (f_0, f_1, \dots, f_{2n-1}) \in \F_q^{2n}$ its coefficient vector. This notation will be used to relate cyclic codes and their associated additive conjucyclic codes via the isomorphism $\Psi_\beta$ introduced in \eqref{eq2-exp}.

\begin{theorem}\label{thm:generator}
	Let \(\mathcal{C}\) be an additive conjucyclic code over \(\F_{q^2}\) of length \(n\) and let \(\mathcal{D} = \Psi_\beta(\mathcal{C}) = \langle g(x) \rangle\) be the corresponding \(q\)-ary linear cyclic code of length \(2n\), where \(g(x) \in \Div_{\F_q}(x^{2n}-1)\) and \(\deg(g(x)) = k\). Denote by \(\pmb{g} = (g_0, g_1, \dots, g_{2n-1})\) the coefficient vector of \(g(x)\). Then a generator matrix of \(\mathcal{C}\) is given by
	\[
	\hat{G} = \begin{pmatrix}
		\Psi_\beta^{-1}(\pmb{g}) \\
		T(\Psi_\beta^{-1}(\pmb{g})) \\
		\vdots \\
		T^{2n-k-1}(\Psi_\beta^{-1}(\pmb{g}))
	\end{pmatrix},
	\]
	where \(T\) is the right conjucyclic shift operator defined in \eqref{eq222xx}.
\end{theorem}

\begin{proof}
	We first show that every codeword \(\mathbf{c} \in \mathcal{C}\) lies in the row space of \(\hat{G}\). Since \(\mathcal{D} = \langle g(x) \rangle\) is a cyclic code, its generator matrix is given by
	\[
	G_{\mathcal{D}} = \begin{pmatrix}
		\pmb{g} \\
		\sigma(\pmb{g}) \\
		\vdots \\
		\sigma^{2n-k-1}(\pmb{g})
	\end{pmatrix},
	\]
	where \(\pmb{g} = (g_0, g_1, \dots, g_{2n-1})\) is the coefficient vector of \(g(x)\). So there exist scalars \(k_0, k_1, \dots, k_{2n-k-1} \in \F_q\) such that
	\[
	\Psi_\beta(\mathbf{c}) = \sum_{i=0}^{2n-k-1} k_i \sigma^i(\pmb{g}).
	\]
	Applying the inverse isomorphism \(\Psi_\beta^{-1}\) and using Proposition~\ref{prop:commute} which gives \(\Psi_\beta^{-1}(\sigma^i(\pmb{g})) = T^i(\Psi_\beta^{-1}(\pmb{g}))\), we obtain
	\[
	\mathbf{c} = \sum_{i=0}^{2n-k-1} k_i T^i(\Psi_\beta^{-1}(\pmb{g})).
	\]
	Thus each codeword is a linear combination of the rows of \(\hat{G}\). Moreover, since \(|\mathcal{C}| = |\mathcal{D}| = q^{2n-k}\) and the rows of \(\hat{G}\) are linearly independent over \(\F_q\) (as they are images of the basis of \(\mathcal{D}\)), they form an \(\F_q\)-basis of \(\mathcal{C}\). Consequently, \(\hat{G}\) is a generator matrix of \(\mathcal{C}\).
\end{proof}

\begin{example}\label{exam111}
	Let \(q = 4\) and \(n = 11\). Over \(\F_4\), the complete factorization of \(x^{22}-1\) is
	\[
	(1+x)^2 (1+\omega^2 x + x^2 + x^3 +\omega x^4 + x^5)^2 (1+\omega x + x^2 + x^3 +\omega^2 x^4 + x^5)^2,
	\]
	where \(\omega\) is a primitive element of \(\F_4\). Hence every divisor of \(x^{22}-1\) is of the form \((1+x)^{i_1}(1+\omega^2 x + x^2 + x^3 +\omega x^4 + x^5)^{i_2}(1+\omega x + x^2 + x^3 +\omega^2 x^4 + x^5)^{i_3}\) with \(0\le i_1,i_2,i_3\le 2\). By Theorem~\ref{theo3.7}, there are \(3^3 = 27\) distinct additive conjucyclic codes over \(\F_{16}\) of length \(11\).
	
	Consider the code \(\mathcal{C}\) determined by the divisor
	\[
	g(x) = (1+\omega^2 x + x^2 + x^3 +\omega x^4 + x^5)^2 = 1+\omega x^2 + x^4 + x^6 +\omega^2 x^8 + x^{10}.
	\]
	Let \(\pmb{g} = (g_0,\dots,g_{21})\) be the coefficient vector of \(g(x)\); explicitly,
	\[
	\pmb{g} = (1,0,\omega,0,1,0,1,0,\omega^2,0,1,0,0,0,0,0,0,0,0,0,0,1)\in\F_4^{22}.
	\]
	Choose a primitive element \(\beta\) of \(\F_{16}\) and set \(\mathbf{w} = \Psi_\beta^{-1}(\pmb{g})\). Then
	\[
	\mathbf{w} = (\beta,0,\beta^{6},0,\beta,0,\beta,0,\beta^{11},0,\beta)\in\F_{16}^{11}.
	\]
	According to Theorem~\ref{thm:generator}, a generator matrix of \(\mathcal{C}\) is
	\[
	\hat{G} = \begin{pmatrix}
		\mathbf{w} \\
		T(\mathbf{w}) \\
		\vdots \\
		T^{2n-k-1}(\mathbf{w})
	\end{pmatrix},
	\]
	where \(k = \deg(g(x)) = 10\) and thus \(2n-k = 12\). Using Magma~\cite{bosma1997}, we find that \(\mathcal{C}\) is an \((11,4^{12},5)_{16}\) additive code; its weight distribution is
	\[
	1 + 825z^{5} + 1980z^{6} + 61875z^{7} + 391875z^{8} + 2025375z^{9} + 6045600z^{10} + 8249685z^{11}.
	\]
	The Singleton bound for additive codes over \(\F_{q^2}\) (see, e.g., \cite{huffman2003}) gives \(d \le n - \log_{q^2}(M) + 1 = 11 - 6 + 1 = 6\). Since the code attains \(d = 5\), it is a near MDS additive code, i.e., it misses the Singleton bound by only one.
\end{example}

\section{Alternating duality and quantum code Construction}\label{sec:generator}

For a $q^2$-ary additive code $\mathcal{C}$, both $\mathcal{C}^{\perp_e}$ and $\mathcal{C}^{\perp_s}$ are $\F_{q^2}$-linear because the Euclidean and symplectic inner products satisfy $\langle \mathbf{u}, a\mathbf{v} \rangle = a \langle \mathbf{u}, \mathbf{v} \rangle$ for any $a \in \F_{q^2}$. Thus, if $\mathbf{v}$ is orthogonal to $\mathcal{C}$, so is any scalar multiple $a\mathbf{v}$. This makes these conventional inner products ill‑suited for studying the dual structure of additive codes.

To remedy this, we introduce an alternating inner product on $\F_{q^2}^n$ that is more suitable for additive codes. We then show that the mapping $\Psi_\beta$ defined in \eqref{eq2-exp} preserves orthogonality between the symplectic inner product on $\F_q^{2n}$ and the alternating inner product on $\F_{q^2}^n$. 



\begin{proposition}\label{prop:altinner}
	Let $\beta$ be a primitive element of $\F_{q^2}$. For $\pmb{u} = (u_0,\ldots,u_{n-1})$ and $\pmb{v} = (v_0,\ldots,v_{n-1})$ in $\F_{q^2}^n$, we define
	\begin{equation}\label{rrrrq123}
		\langle \pmb{u}, \pmb{v} \rangle_{a,\beta} = (\bar{\beta}^2 - \beta^2) \sum_{i=0}^{n-1} (u_i \bar{v}_i - \bar{u}_i v_i).
	\end{equation}
	Then $\langle \cdot,\cdot \rangle_{a,\beta}$ is a nondegenerate alternating inner product on $\F_{q^2}^n$. In other words, it satisfies the following properties:
	\begin{enumerate}
		\item $\langle \pmb{u}, \pmb{v} \rangle_{a,\beta} \in \F_q$ for all $\pmb{u},\pmb{v}$;
		\item $\langle \pmb{u}, \pmb{v} \rangle_{a,\beta}$ is $\F_q$-bilinear;
		\item $\langle \pmb{u}, \pmb{v} \rangle_{a,\beta} = -\langle \pmb{v}, \pmb{u} \rangle_{a,\beta}$ (alternating);
		\item $\langle \pmb{u}, \pmb{u} \rangle_{a,\beta} = 0$ for all $\pmb{u}$;
		\item If $\langle \pmb{u}, \pmb{v} \rangle_{a,\beta} = 0$ for all $\pmb{v} \in \F_{q^2}^n$, then $\pmb{u} = \mathbf{0}$ (nondegeneracy).
	\end{enumerate}
\end{proposition}

\begin{proof}
	Set $\gamma = \bar{\beta}^2 - \beta^2$. Observe that $\gamma \neq 0$ because $\beta \notin \F_q$, and $\gamma^q = \bar{\beta}^{2q} - \beta^{2q} = \beta^{2q^2} - \bar{\beta}^{2q} = \beta^{2} - \bar{\beta}^{2} = -\gamma$. If $\operatorname{char}(\F_{q^2}) = 2$, then $-\gamma = \gamma$, so $\gamma^q = \gamma$; otherwise $(\gamma^q)^2 = \gamma^{2q} = (\gamma^2)^q = \gamma^2$, which also forces $\gamma^q = \gamma$. Hence $\gamma \in \F_q \setminus \{0\}$.
	
	\begin{enumerate}
		\item For each $i$, the quantity $u_i\bar{v}_i - \bar{u}_i v_i$ is equal to its own $q$-th power, therefore lies in $\F_q$. Multiplying by $\gamma \in \F_q$ and summing over $i$ yields an element of $\F_q$.
		
		\item The map $\pmb{v} \mapsto \langle \pmb{u}, \pmb{v} \rangle_{a,\beta}$ is a linear combination of the coordinates of $\pmb{v}$ with coefficients in $\F_q$, hence $\F_q$-linear. By symmetry (up to sign), the same holds for the first argument, so the form is $\F_q$-bilinear.
		
		\item Direct computation gives
		\[
		\langle \pmb{v}, \pmb{u} \rangle_{a,\beta} = \gamma \sum_i (v_i\bar{u}_i - \bar{v}_i u_i) = -\gamma \sum_i (u_i\bar{v}_i - \bar{u}_i v_i) = -\langle \pmb{u}, \pmb{v} \rangle_{a,\beta}.
		\]
		
		\item Taking $\pmb{v} = \pmb{u}$ in the previous property yields $\langle \pmb{u}, \pmb{u} \rangle_{a,\beta} = 0$.
		
		\item Assume $\langle \pmb{u}, \pmb{v} \rangle_{a,\beta} = 0$ for all $\pmb{v} \in \F_{q^2}^n$. For a fixed coordinate $j$, choose $\pmb{v}$ with $v_j = 1$ and all other entries $0$. Then
		\[
		\gamma (u_j \bar{1} - \bar{u}_j \cdot 1) = \gamma (u_j - \bar{u}_j) = 0,
		\]
		so $u_j = \bar{u}_j$, i.e., $u_j \in \F_q$. Next, choose $\pmb{v}$ with $v_j = \beta$ and $0$ elsewhere. We obtain
		\[
		\gamma (u_j \bar{\beta} - \bar{u}_j \beta) = \gamma (u_j \beta^q - u_j \beta) = \gamma u_j (\beta^q - \beta) = 0.
		\]
		Since $\gamma \neq 0$ and $\beta^q - \beta \neq 0$ (as $\beta \notin \F_q$), it follows that $u_j = 0$. This holds for every $j$, hence $\pmb{u} = \mathbf{0}$.
	\end{enumerate}
\end{proof}

\begin{remark}
	In the special case $q = 2$, we have $\beta = 1$ and the alternating inner product $\langle \cdot,\cdot \rangle_{a,\beta}$ defined in \eqref{rrrrq123} reduces to
	\[
	\langle \pmb{u}, \pmb{v} \rangle_{a,\beta} = \sum_{i=1}^{n} (u_i \bar{v}_i + \bar{u}_i v_i) = \Tr\!\left( \sum_{i=1}^{n} u_i \bar{v}_i \right).
	\]
	Hence, the alternating inner product can be viewed as a generalization of the classical trace-Hermitian inner product introduced in \cite{calderbank1998}.
\end{remark}


\begin{proposition}\label{prop:orth}
	For any $\mathbf{u}, \mathbf{v} \in \F_{q^2}^n$, we have
	\[
	\langle \Psi_\beta(\mathbf{u}), \Psi_\beta(\mathbf{v}) \rangle_s = \langle \mathbf{u}, \mathbf{v} \rangle_{a,\beta},
	\]
	where $\langle \cdot,\cdot \rangle_s$ and $\langle \cdot,\cdot \rangle_{a,\beta}$ are the symplectic inner product on $\F_q^{2n}$ (see \eqref{wewe12}) and the alternating inner product on $\F_{q^2}^n$ (see \eqref{rrrrq123}), respectively.
\end{proposition}
\begin{proof}
	Write $\mathbf{u} = (u_0,\ldots,u_{n-1})$ and $\mathbf{v} = (v_0,\ldots,v_{n-1})$. Then
	\[
	\Psi_\beta(\mathbf{u}) = \bigl(\Tr(\beta u_0),\ldots,\Tr(\beta u_{n-1}),\Tr(\bar{\beta}u_0),\ldots,\Tr(\bar{\beta}u_{n-1})\bigr),
	\]
	and similarly for $\Psi_\beta(\mathbf{v})$. Using the definition of the symplectic inner product,
	\[
	\begin{aligned}
		\langle\Psi_\beta(\mathbf{u}),\Psi_\beta(\mathbf{v})\rangle_s
		&= \sum_{i=0}^{n-1}\Bigl(\Tr(\bar{\beta}u_i)\Tr(\beta v_i) - \Tr(\beta u_i)\Tr(\bar{\beta}v_i)\Bigr)\\
		&= \sum_{i=0}^{n-1}\Bigl((\bar{\beta}u_i+\beta\bar{u}_i)(\beta v_i+\bar{\beta}\bar{v}_i) - (\beta u_i+\bar{\beta}\bar{u}_i)(\bar{\beta}v_i+\beta\bar{v}_i)\Bigr)\\
		&= \sum_{i=0}^{n-1}\bigl(\bar{\beta}^2 u_i\bar{v}_i + \beta^2\bar{u}_i v_i - \beta^2 u_i\bar{v}_i - \bar{\beta}^2\bar{u}_i v_i\bigr)\\
		&= (\bar{\beta}^2-\beta^2)\sum_{i=0}^{n-1}(u_i\bar{v}_i-\bar{u}_i v_i) = \langle\mathbf{u},\mathbf{v}\rangle_{a,\beta},
	\end{aligned}
	\]
which completes the proof. 
\end{proof}

\begin{definition}\label{def:altdual}
	The alternating dual code of $\mathcal{C}$ is defined as
	\[
	\mathcal{C}^{\perp_a} = \{ \pmb{v} \in \F_{q^2}^n \mid \langle \pmb{u}, \pmb{v} \rangle_{a,\beta} = 0 \text{ for all } \pmb{u} \in \mathcal{C} \}.
	\]
\end{definition}

\begin{theorem}\label{sd22323}
	Let $\mathcal{C}$ be an additive conjucyclic code over $\F_{q^2}$ of length $n$ and $\mathcal{D} = \Psi(\mathcal{C})$ be a $q$-ary linear cyclic code of length $2n$. Then we have
	\[
	\mathcal{C}^{\perp_a} = \Psi^{-1}(\mathcal{D}^{\perp_s}).
	\]
\end{theorem}
\begin{proof}
	For any $\mathbf{c}\in\mathcal{C}$ and $\mathbf{d}\in\mathcal{C}^{\perp_a}$, by Proposition \ref{prop:orth}, we have $\langle\mathbf{c},\mathbf{d}\rangle_a = \langle\Psi(\mathbf{c}),\Psi(\mathbf{d})\rangle_s = 0$. Therefore, $\Psi(\mathbf{d})\in\mathcal{D}^{\perp_s}$, i.e., $\Psi(\mathcal{C}^{\perp_a})\subseteq\mathcal{D}^{\perp_s}$ and $\mathcal{C}^{\perp_a}\subseteq\Psi^{-1}(\mathcal{D}^{\perp_s})$. Similarly, we can also prove $\Psi^{-1}(\mathcal{D}^{\perp_s})\subseteq\mathcal{C}^{\perp_a}$. This concludes the result.
\end{proof}

Let $g(x) \in \Div_{\F_q}(x^{2n}-1)$ with $\deg(g(x)) = k$, and set $h(x) = \frac{x^{2n}-1}{g(x)}$. Define the reciprocal polynomial
\[
h^*(x) = \frac{1}{h_{2n-k}} x^{2n-k} h\!\left(\frac{1}{x}\right).
\]
 Denote by $\pmb{h^*} = (h_0^*, h_1^*, \dots, h_{2n-1}^*) \in \F_q^{2n}$ the coefficient vector of $h^*(x)$.

\begin{lemma}\label{lem:dualgen}
	Let $\mathcal{D} = \langle g(x) \rangle$ be a $q$-ary linear cyclic code of length $2n$. Then a generator matrix of the symplectic dual code $\mathcal{D}^{\perp_s}$ is given by
	\[
	H_s = \begin{pmatrix}
		\tau(\pmb{h^*}) \\
		\tau(\sigma(\pmb{h^*})) \\
		\vdots \\
		\tau(\sigma^{k-1}(\pmb{h^*}))
	\end{pmatrix},
	\]
	where $\sigma$ is the cyclic shift operator defined in \eqref{eq111xx} and $\tau$ is the linear transformation defined by
	\[
	\tau(\mathbf{v}) = (v_0, v_1, \dots, v_{2n-1}) \begin{pmatrix} 0 & I_n \\ -I_n & 0 \end{pmatrix} = (-v_n, \dots, -v_{2n-1}, v_0, \dots, v_{n-1})
	\]
	for any $\mathbf{v} = (v_0, v_1, \dots, v_{2n-1}) \in \F_q^{2n}$, with $I_n$ the $n \times n$ identity matrix.
\end{lemma}
\begin{proof}
	By the theory of cyclic codes \cite{huffman2003}, the Euclidean dual code $\mathcal{D}^{\perp_e}$ is generated by the matrix
	\[
	H_e = \begin{pmatrix}
		\pmb{h^*} \\
		\sigma(\pmb{h^*}) \\
		\vdots \\
		\sigma^{k-1}(\pmb{h^*})
	\end{pmatrix}.
	\]
	For any codeword $\mathbf{d} \in \mathcal{D}$ and any $i \in \{0, \dots, k-1\}$, we have
	\[
	\langle \mathbf{d}, \tau(\sigma^i(\pmb{h^*})) \rangle_s = \langle \mathbf{d}, \sigma^i(\pmb{h^*}) \rangle_e = 0.
	\]
	Hence, the rows of $H_s$ belong to $\mathcal{D}^{\perp_s}$. Moreover, since $H_s = H_e \begin{pmatrix} 0 & I_n \\ -I_n & 0 \end{pmatrix}$, we have that $\operatorname{rank}(H_s) = \operatorname{rank}(H_e) = k$. Therefore, $H_s$ is a generator matrix of $\mathcal{D}^{\perp_s}$.
\end{proof}

\begin{theorem}\label{thm:parity}
	Let $\mathcal{C}$ be a $q^2$-ary additive conjucyclic code of length $n$ corresponding to the $q$-ary linear cyclic code $\mathcal{D} = \langle g(x) \rangle$, where $g(x) \in \Div_{\F_q}(x^{2n}-1)$ and $\deg(g(x)) = k$. Then a generator matrix of the alternating dual code $\mathcal{C}^{\perp_{a,\beta}}$ is given by
	\begin{equation}\label{rere11}
		\hat{H}_a = \begin{pmatrix}
		\Psi_\beta^{-1}(\tau(\pmb{h^*})) \\
		\Psi_\beta^{-1}(\tau(\sigma(\pmb{h^*}))) \\
		\vdots \\
		\Psi_\beta^{-1}(\tau(\sigma^{k-1}(\pmb{h^*})))
	\end{pmatrix},
	\end{equation}
	where $\pmb{h^*} \in \F_q^{2n}$ is the coefficient vector of $h^*(x)$, $\sigma$ is the cyclic shift operator, and $T$ is the right conjucyclic shift operator defined in \eqref{eq222xx}.
\end{theorem}
\begin{proof}
	By Theorem~\ref{sd22323}, we have $\mathcal{C}^{\perp_{a,\beta}} = \Psi_\beta^{-1}(\mathcal{D}^{\perp_s})$. 
	Lemma~\ref{lem:dualgen} provides a generator matrix $H_s$ of $\mathcal{D}^{\perp_s}$ whose rows are $\tau(\sigma^i(\pmb{h^*}))$ for $i=0,\dots,k-1$. 
	Since $\Psi_\beta^{-1}$ is an $\F_q$-linear isomorphism, applying it to each row of $H_s$ yields a generator matrix of $\Psi_\beta^{-1}(\mathcal{D}^{\perp_s}) = \mathcal{C}^{\perp_{a,\beta}}$, namely the matrix $\hat{H}_a$ as defined above. Hence $\hat{H}_a$ is a generator matrix of $\mathcal{C}^{\perp_{a,\beta}}$.
\end{proof}

\begin{remark}
	The generator matrix $\hat{H}_a$ of $\mathcal{C}^{\perp_{a,\beta}}$ can also be regarded as a parity-check matrix for $\mathcal{C}$. Indeed, $\mathcal{C}$ is completely characterized by $\hat{H}_a$ via the condition
	\[
	\mathcal{C} = \left\{ \mathbf{c} \in \F_{q^2}^n \mid \hat{H}_a \cdot \mathbf{c}^\top = \mathbf{0} \right\},
	\]
	where the product is taken with respect to the alternating inner product, i.e., each row of $\hat{H}_a$ is orthogonal to every codeword under $\langle \cdot, \cdot \rangle_{a,\beta}$.
\end{remark}

\begin{corollary}\label{rere113}
	If $\charac(\F_{q^2}) = 2$, then the generator matrix $\hat{H}_a$ of $\mathcal{C}^{\perp_{a,\beta}}$ in \eqref{rere11} can be simplified as
	\[
	\hat{H}_a = \begin{pmatrix}
		\Psi_\beta^{-1}(\tau(\pmb{h^*})) \\
		T(\Psi_\beta^{-1}(\tau(\pmb{h^*}))) \\
		\vdots \\
		T^{k-1}(\Psi_\beta^{-1}(\tau(\pmb{h^*})))
	\end{pmatrix}.
	\]
\end{corollary}
\begin{proof}
	By Theorem~\ref{thm:parity}, the rows of $\hat{H}_a$ are $\Psi_\beta^{-1}(\tau(\sigma^i(\pmb{h^*})))$ for $i = 0,\dots,k-1$. When $\charac(\F_{q^2}) = 2$, a direct computation shows that $\tau$ commutes with the cyclic shift $\sigma$; i.e., $\tau(\sigma(\mathbf{v})) = \sigma(\tau(\mathbf{v}))$ for any $\mathbf{v}\in\F_q^{2n}$. Hence $\tau(\sigma^i(\pmb{h^*})) = \sigma^i(\tau(\pmb{h^*}))$ for all $i$. Applying Proposition~\ref{prop:commute}, we have $\Psi_\beta^{-1}(\sigma^i(\tau(\pmb{h^*}))) = T^i(\Psi_\beta^{-1}(\tau(\pmb{h^*})))$. Therefore
	\[
	\Psi_\beta^{-1}(\tau(\sigma^i(\pmb{h^*}))) = T^i(\Psi_\beta^{-1}(\tau(\pmb{h^*}))),
	\]
	and the claimed simplified form follows.
\end{proof}

\begin{corollary}\label{cor:char2}
	Let $\mathcal{C}$ be a $q^2$-ary additive conjucyclic code. If $\charac(\F_{q^2}) = 2$, then the alternating dual code $\mathcal{C}^{\perp_{a,\beta}}$ is also additive conjucyclic.
\end{corollary}
\begin{proof}
 When $\charac(\F_{q^2}) = 2$, Corollary~\ref{rere113} shows that the alternating dual code $\mathcal{C}^{\perp_{a,\beta}}$ admits a generating set consisting of the vectors $T^i(\mathbf{w})$ for $i = 0,\dots,k-1$, where $\mathbf{w} = \Psi_\beta^{-1}(\tau(\pmb{h^*}))$. Since $T$ is the right conjucyclic shift operator, the $\F_q$-linear span of these vectors is closed under $T$. Consequently, $\mathcal{C}^{\perp_{a,\beta}}$ is an additive conjucyclic code.
\end{proof}

\begin{remark}
	In \cite{abualrub2020}, it was shown that for a quaternary additive conjucyclic code $\mathcal{C}$, its trace dual code
	\[
	\mathcal{C}^{\mathrm{Tr}} = \{\mathbf{v} \in \F_4^n \mid \Tr(\langle \mathbf{u}, \mathbf{v} \rangle_e) = 0 \text{ for all } \mathbf{u} \in \mathcal{C}\}
	\]
	is also additive conjucyclic. This result can now be seen as a special case of Corollary~\ref{cor:char2}. Indeed, when $q = 2$, the alternating inner product reduces to the trace‑Hermitian form, and one verifies that $\mathcal{C}^{\mathrm{Tr}} = (\mathcal{C}^{\perp_{a,\beta}})^2$, where $(\cdot)^2$ denotes componentwise squaring. Since Corollary~\ref{cor:char2} guarantees that $\mathcal{C}^{\perp_{a,\beta}}$ is additive conjucyclic, the same property follows immediately for $\mathcal{C}^{\mathrm{Tr}}$. Thus the present work not only generalizes the known result to arbitrary $q$ but also places it within a unified algebraic framework.
\end{remark}

The connection between classical codes and QECCs is well established via the stabilizer formalism, where the symplectic inner product serves as the key link \cite{ashikhmin2001, ketkar2006}.

\begin{lemma}\cite{ashikhmin2001}\label{ss113r}
	If $\mathcal{C}$ is a symplectic dual-containing linear code with parameters $[2n,k]_q$, then there exists an $[[n,k-n,\ge w_s(\mathcal{C})]]_q$ QECC that is pure to $w_s(\mathcal{C})$.
\end{lemma}

As a direct consequence of the established correspondence and the properties of the alternating inner product, we obtain the following quantum code construction.
\begin{theorem}\label{thm:qeck}
	Let $\mathcal{C}$ be an $(n,M=q^k,w_h(\mathcal{C}))_{q^2}$ additive conjucyclic code satisfying $\mathcal{C}^{\perp_{a,\beta}} \subseteq \mathcal{C}$. Then there exists a pure $[[n,k-n,\ge w_h(\mathcal{C})]]_q$ QECC.
\end{theorem}
\begin{proof}
	By Theorem~\ref{thm:correspondence}, the code $\mathcal{C}$ corresponds to a $q$-ary linear cyclic code $\mathcal{D} = \Psi_\beta(\mathcal{C})$ of length $2n$ with $|\mathcal{D}| = q^k$. Moreover, Theorem~\ref{thm:correspondence} also gives $w_h(\mathcal{C}) = w_s(\mathcal{D})$. Since $\mathcal{C}^{\perp_{a,\beta}} \subseteq \mathcal{C}$, Proposition~\ref{prop:orth} implies $\mathcal{D}^{\perp_s} \subseteq \mathcal{D}$. Thus $\mathcal{D}$ is a symplectic dual-containing linear code with parameters $[2n,k]_q$. Applying Lemma~\ref{ss113r} yields a pure $[[n,k-n,\ge w_s(\mathcal{D})]]_q = [[n,k-n,\ge w_h(\mathcal{C})]]_q$ QECC, as desired.
\end{proof}

\begin{example}\label{ex:dual}
	We continue with the same parameters as in Example~\ref{exam111}: let \(q = 4\), \(n = 11\), and let \(\mathcal{C}\) be the additive conjucyclic code over \(\F_{16}\) determined by the divisor
	\[
	g(x) = (1+\omega^2 x + x^2 + x^3 +\omega x^4 + x^5)^2 = 1+\omega x^2 + x^4 + x^6 +\omega^2 x^8 + x^{10},
	\]
	where \(\omega\) is a primitive element of \(\F_4\). With the same primitive element \(\beta\) of \(\F_{16}\) as before, we now examine the dual side.
	
	Set \(h(x) = \frac{x^{22}-1}{g(x)} = 1 + \omega x^{2} + \omega x^{4} + \omega^{2}x^{8} + \omega^{2}x^{10} + x^{12}\). Then
	\[
	h^{*}(x) = 1 + \omega^{2}x^{2} + \omega^{2}x^{4} + \omega x^{8} + \omega x^{10} + x^{12},
	\]
	and its coefficient vector is
	\[
	\pmb{h^*} = (1,0,\omega^{2},0,\omega^{2},0,0,0,\omega,0,\omega,1,0,0,0,0,0,0,0,0,0,0)\in\F_4^{22}.
	\]
	Applying the linear transformation \(\tau\) from Lemma~\ref{lem:dualgen} yields
	\[
	\tau(\pmb{h^*}) = (0,1,0,0,0,0,0,0,0,0,1,0,\omega^{2},0,\omega^{2},0,0,0,\omega,0,\omega,0)\in\F_4^{22}.
	\]
	Set \(\mathbf{u} = \Psi_\beta^{-1}(\tau(\pmb{h^*}))\); explicitly,
	\[
	\mathbf{u} = (\beta^{4},\beta,\beta^{14},0,\beta^{14},0,0,0,\beta^{9},0,\beta^{9})\in\F_{16}^{11}.
	\]
	Since \(\charac(\F_{16})=2\), 
	by Corollary~\ref{cor:char2}, a generator matrix of the alternating dual code \(\mathcal{C}^{\perp_{a,\beta}}\) is
	\[
	\hat{H}_a = \begin{pmatrix}
		\mathbf{u} \\
		T(\mathbf{u}) \\
		\vdots \\
		T^{9}(\mathbf{u})
	\end{pmatrix}.
	\]

	One verifies directly that \(\mathcal{C}^{\perp_{a,\beta}} \subseteq \mathcal{C}\); thus \(\mathcal{C}\) is alternating dual‑containing. Applying Theorem~\ref{thm:qeck} we obtain a pure \([[11,1,5]]_4\) QECC. Its weight enumerator is
	\[
	825z^{5} + 1155z^{6} + 61050z^{7} + 361350z^{8} + 91904925z^{9} + 5664615z^{10} + 7734720z^{11}.
	\]
This code not only exceeds the quantum Gilbert–Varshamov bound \cite{feng2004} and improves upon the previously known \([[19,1,5]]_4\) code listed in \cite{edel}, but is in fact optimal according to Grassl's code tables \cite{Grassl}. Moreover, unlike the optimal stabilizer codes recorded in the tables—whose stabilizer matrices typically lack algebraic structure and thus require large storage—our construction yields a generator matrix with a cyclic structure, which significantly reduces the memory requirement.
\end{example}

\section{Conclusions}\label{sec:conclusion}
In this paper, we have developed the algebraic theory of additive conjucyclic codes over $\F_{q^2}$. We first established a one-to-one correspondence between $q^2$-ary additive conjucyclic codes and $q$-ary linear cyclic codes via the trace map, which allows us to determine the enumeration of such codes and provide explicit forms of their generator matrices. We then introduced an alternating inner product on $\F_{q^2}^n$ and, under this inner product, derived a necessary and sufficient condition for these codes to be dual-containing, from which we obtained explicit parity-check matrices and a construction method for $q$-ary QECCs. This work provides a systematic treatment of additive conjucyclic codes over general $\F_{q^2}$ and their application to non-binary quantum error correction. We hope that the framework established herein will inspire further research on constructing QECCs with good parameters from additive conjucyclic codes.

\section*{Author contributions}
J. Lv: Writing – original draft, formal analysis, validation; X. Lian: Conceptualization, methodology, investigation; R. Li: Conceptualization, funding acquisition; H. Hou: Data curation, formal analysis, funding acquisition, supervision. All authors have read and approved the final manuscript.






\section*{Use of Generative-AI tools declaration}
The authors declare they have not used Artificial Intelligence (AI) tools in the creation of this article.

\section*{Acknowledgments}
This research is supported in part by National Key Research and Development Program of China under Grant No.
2025YFA1017200, the National Natural Science Foundation of China under Grant No. 62401144, the Natural Science Foundation
of Guangdong Province under Grant No. 2026A1515012796.

\section*{Conflict of interest}
The authors declare no conflict of interest.

For more questions regarding reference style, please refer to the \href{http://www.ncbi.nlm.nih.gov/books/NBK7256/}{Citing Medicine}.

\end{document}